\documentclass[smallextended]{svjour3} 
\pdfoutput=1
\smartqed  
\usepackage{amsfonts, amsmath, amssymb, alltt, color, pbox, enumerate, comment}
\usepackage{algorithm}
\usepackage[pdftex]{graphicx}
\usepackage{tabularx}
\usepackage[title]{appendix}
\usepackage{mathtools}
\usepackage{xspace}
\usepackage{todonotes}

\newcommand{\dpolylog}{DPolylogTime\xspace}
\newcommand{\polylog}{\mathrm{DPolylogTime}}
\newcommand{\npolylog}{NPolylogTime\xspace}
\newcommand{\mnpolylog}{\mathrm{NPolylogTime}}
\newcommand{\dtime}[1]{\mathrm{DTIME}(#1)}
\newcommand{\ntime}[1]{\mathrm{NTIME}(#1)}
\newcommand{\atime}[2]{\mathrm{ATIME}(#1,#2)}
\newcommand{\atimeop}[2]{\mathrm{ATIME}^{\mathit{op}}(#1,#2)}
\newcommand{\mA}{\mathbf{A}}
\newcommand{\mB}{\mathbf{B}}
\newcommand{\mS}{\mathbf{S}}
\newcommand{\initial}[1]{\mathrm{InitialZeros}^{#1}}
\newcommand{\conseq}[1]{\mathrm{ConseqZeros}^{#1}}
\newcommand{\noconseq}[1]{\mathrm{NoConseqZeros}^{#1}}
\newcommand{\exactlyonce}[1]{\mathrm{ExactlyOnce}^{#1}}
\newcommand{\genconseq}[2]{\mathrm{AtLeastBlocks}^{#1}_{#2}}
\newcommand{\genexactly}[2]{\mathrm{ExactlyBlocks}^{#1}_{#2}}
\newcommand{\soplog}{\mathrm{SO}^{\mathit{plog}}}
\newcommand{\sigmaplog}[1]{\Sigma^\mathit{plog}_{#1}}
\newcommand{\piplog}[1]{\Pi^\mathit{plog}_{#1}}
\newcommand{\polysigma}[1]{\tilde{\Sigma}_{#1}^{\mathit{plog}}}
\newcommand{\polypi}[1]{\tilde{\Pi}_{#1}^{\mathit{plog}}}
\newcommand{\ttx}{{\tt x}}
\newcommand{\tty}{{\tt y}}

\newcommand{\polylogspace}{\mathrm{PolylogSpace}}

\begin{document}

\title{Completeness in Polylogarithmic Time and Space\thanks{The research reported in this paper results from the project {\em Higher-Order Logics and Structures} supported by the Austrian Science Fund (FWF: \textbf{[I2420-N31]}). It has also been partly supported by the Austrian Ministry for Transport, Innovation and Technology, the Federal Ministry for Digital and Economic Affairs, and the Province of Upper Austria in the frame of the COMET center SCCH.}}



\author{Flavio Ferrarotti \and Sen\'{e}n Gonz\'{a}lez \and Klaus-Dieter Schewe \and Jos\'{e} Mar\'{\i}a Turull-Torres}


\institute{F. Ferrarotti, S. Gonz\'{a}lez \at
              Software Competence Center Hagenberg, Austria\\
              \email{\{flavio.ferrarotti,senen.gonzalez\}@scch.at}
           \and
           K.-D. Schewe \at
              Zhejiang University, UIUC Institute, Haining, China\\
              \email{kd.schewe@intl.zju.edu.cn}
		   \and 
		   J. M. Turull-Torres \at
		   	Universidad Nacional de La Matanza, Buenos Aires, Argentina\\
			\email{jturull@unlam.edu.ar}
}

\date{Received: date / Accepted: date}

\maketitle

\begin{abstract}
Complexity theory can be viewed as the study of the relationship between computation and applications, understood the former as complexity classes and the latter as problems. Completeness results are clearly central to that view. Many natural algorithms resulting from current applications have polylogarithmic time (PolylogTime) or space complexity (PolylogSpace). The classical Karp notion of complete problem however does not plays well with these complexity classes. It is well known that PolylogSpace does not have complete problems under logarithmic space many-one reductions. In this paper we show similar results for deterministic and non-deterministic PolylogTime as well as for every other level of the polylogarithmic time hierarchy. We achieve that  by following a different strategy based on proving the existence of proper hierarchies of problems inside each class. We then develop an alternative notion of completeness inspired by the concept of uniformity from circuit complexity and prove the existence of a (uniformly) complete problem for PolylogSpace under this new notion. As a consequence of this result we get that complete problems can still play an important role in the study of the interrelationship between polylogarithmic and other classical complexity classes.

\end{abstract}

\section{Introduction}

The complexity theory of polylogarithmic time and space computations has not received as much attention as we believe it deserves. This is the case despite the fact that such computations appear rather naturally. Take for instance dynamic graph connectivity algorithms~\cite{HolmLT98,KapronKM13}, algorithms for updates in the minimum cut problem~\cite{GoranciHT18}, algorithms for maintaining a dynamic family of sequences under equality tests~\cite{MehlhornSU94}, and distance and point-set algorithms in computational geometry~\cite{Chan10,KapoorS96,Smid92,Supowit90}, among many others. Indeed, from a complexity theory perspective the main antecedent that we can think of is the work on constant-depth quasi-polynomial size AND/OR-circuits in~\cite{barrington:sct1992} where it was proven that the class of Boolean queries computable by the class of $\mathrm{DTIME}[(\log n)^{O(1)}]$ DCL-uniform families of Boolean circuits of unbounded fan-in, size $2^{{(\log n)}^{O(1)}}$ and depth $O(1)$ coincides with the class of Boolean queries expressible in a restricted fragment of  second-order logic. The complexity class $\mathrm{DTIME}[2^{(\log n)^{O(1)}}]$ is known as {\em quasipolynomial time}. Interstingly, the fastest known algorithm for checking graph isomorphisms is in quasipolynomial time \cite{babai:stoc2016}. 

In \cite{FerrarottiGST18} we started a deeper investigation of the descriptive complexity of sublinear time computations emphasising  complexity classes DPolylogTime and NPolylogTime of decision problems that can be solved deterministically or non-deterministically with a time complexity in $O(\log^k n)$ for some $k$, where $n$ is as usual the size of the input. We extended these complexity classes to a complete hierarchy, the {\em polylogarithmic time hierarchy}, analogous to the polynomial time hierarchy, and for each class $\Sigma_m^{plog}$ or $\Pi_m^{plog}$ ($m \in \mathbb{N}$) in the hierarchy we defined a fragment of semantically restricted second-order logic capturing it \cite{FerrarottiGST19}. While the hierarchy as a whole captures the same class of problems studied in~\cite{barrington:sct1992}, the various classes of the hierarchy provide fine-grained insights into the nature of describing problems decidable in sublinear time. Moreover, in~\cite{FerrarottiGTBV19,FerrarottiGTBV19b} we introduced a novel two-sorted logic that separates the elements
of the input domain from the bit positions needed to address
these elements, proving that the inflationary and partial fixed point variants of that logic  capture DPolylogTime and polylogarithmic space (PolylogSpace), respectively. 

We share the view of complexity as the intricate and exquisite interplay
between computation (complexity classes) and applications (that is, problems)~\cite{Papa07}. Logics are central to this approach since they excel in expressing and capturing computation, but so are completeness results. Following the capture of the main complexity classes of plylogarithmic time and space by means of natural logics, then the obvious question is whether there are complete problems in these complexity classes and what would be an appropriate notion of reduction to define them. In principle, it is well known that PolylogSpace does not have complete problems under logarithmic space many-one reductions, i.e., under the classical Karp notion of completeness. As explained among others in~\cite{Johnson90}, this is due to the space hierarchy theorem by Hartmanis et al.~\cite{HartmanisLS65}.  

Our fisrt attempt to address this problem appeared in~\cite{FerrarottiGST20}. Following an approach inspired by our work in the descriptive complexity of polylogarithmic time, we discovered that there exist proper hierarchies of problems inside each of the classes $\polysigma{m}$ and $\polypi{m}$. As rather straightforward consequence of these proper hierarchies we get that for none of the classes $\polysigma{m}$ and $\polypi{m}$ ($m \in \mathbb{N}$) in the polylogarithmic time hierarchy there exists a complete problem in the classical sense of Karp reductions, i.e., not even under polynomial-time many-one reductions. Moreover, we show that the same holds for DPolylogTime. Note that an approach based instead in the time hierarchy theorem of  Hartmanis et al. cannot be applied to the case of polylogarithmic time, since that hierarchy theorem requires at least linear time. This contrasts with the case of PolylogSpace where the  space hierarchy theorem of  Hartmanis et al. can indeed be applied to show a similar result. 

The results in this paper confirm the ones that appear in the conference version~\cite{FerrarottiGST20} regarding the absence of classical complete problems under Karp reductions inside the different polylogarithmic time complexity classes. Here we abstract from the descriptive complexity considerations made in~\cite{FerrarottiGST20} and concentrate in structural complexity. These results together with the similar well known result regarding PolylogSpace (see~\cite{Johnson90} among others) would in principle mean that these classes are somehow less robust than their polynomial time and space counterparts. In this paper we show that this is not necessarily the case. By considering an alternative notion of completeness, we show that we can still isolate the most difficult problems inside PolylogSpace and draw standard conclusions of the kind entailed by the classical notion of completeness. 

Our alternative notion of completeness (and hardness) is grounded in the concept of uniformity borrowed from circuit complexity theory (see~\cite{Immerman99} and~\cite{BalcazarDG90} among others), hence we call it \emph{uniform completeness}. The intuitive idea is to consider a countably infinite family of problems instead of a single global problem. Each problem in the family corresponding to a fragment of a same global problem determined by a positive integer parameter. Such problem is uniformly complete for a given complexity class if there is a transducer Turing machine which given a positive integer as input builds a Turing machine in the required complexity class that decides the fragment of the problem corresponding to this parameter. In the specific case of PolylogSpace studied here, the machine is a direct-access Turing machine as defined in~\cite{FerrarottiGTBV19b,FerrarottiGTBV19} and the parameter is the exponent affecting the logarithmic function in the space upper bound of any given machine in the class. We choose to use direct-access Turing machines instead of random-access or standard Turing machines simply because we find them easier to work with and are nevertheless equivalent for the problem at hand (see Corollary~1 and Proposition~1 in~\cite{FerrarottiGTBV19b}).

The remainder of this paper is organized as follows. Section~\ref{complexityClasses} summarizes the necessary preliminaries regarding polylogarithmic time complexity classes and its fundamental model of computation based in random-access Turing machines. Section~\ref{section:properhierarchies} is devoted to prove the existence of proper hierarchies of problems in DPolylogTime, NPolylogTime and each additional level of the polylogarithmic time hierarchy. Then the non-existence of classical complete problems for these classes under Karp reductions arises as a rather straightforward consequence, as we show in Section \ref{sec:complete}. At this point we need to introduce some additional preliminaries for our research on an alternative notion of (uniform) completeness. This is done in Sections~\ref{datm} and~\ref{section:polylogspace} were we define and discuss direct-access Turing machines and the $\polylogspace$ complexity class, respectively. We introduce our new, alternative notion of (uniform) completeness in Section~\ref{sec:ucompl}, where we also discuss why this notion is relevant. In Section~\ref{S-plQSAT-C} we define a problem that is uniformly complete for $\mathrm{PolylogSpace}$ as proven in Section~\ref{SS-plQSAT-C}. We conclude with a brief summary in Section \ref{sec:schluss}.

%
%
%
%
%
%
%
%

\section{Polylogarithmic Time}\label{complexityClasses}


The sequential access that Turing machines have to their tapes makes it impossible to compute anything in sublinear time. Therefore, logarithmic time complexity classes are usually studied using models of computation that have random access to their input. As this also applies to the poly-logarithmic complexity classes studied in this paper, we adopt a Turing machine model that has a \emph{random access} read-only input, similar to the log-time Turing machine in~\cite{barrington:jcss1990}.

In the following, $\log n$ always refers to the binary logarithm of $n$, i.e., $\log_2 n$. With $\log^k n$ we mean $(\log n)^k$.

A \emph{random-access Turing machine} is a multi-tape Turing machine with (1) a read-only (random access) \emph{input} of length $n+1$, (2) a fixed number of read-write \emph{working tapes}, and (3) a read-write input \emph{address-tape} of length $\lceil \log n \rceil$.

Every cell of the input as well as every cell of the address-tape contains either $0$ or $1$ with the only exception of the ($n+1$)st cell of the input, which is assumed to contain the endmark $\triangleleft$. In each step the binary number in the address-tape either defines the cell of the input that is read or if this number exceeds $n$, then the ($n+1$)st cell containing $\triangleleft$ is read.   

\begin{example}\label{ex:machine}
Let polylogCNFSAT be the class of satisfiable propositional formulae in conjunctive normal form with $c \leq \lceil \log n \rceil^k$ clauses, where $n$ is the length of the formula. Note that the formulae in polylogCNFSAT tend to have few clauses and many literals. We define a random-access Turing machine $M$ which decides polylogCNFSAT. The alphabet of $M$ is $\{0,1,\#,+,-\}$. The input formula is encoded in the input tape as a list of $c \leq \lceil \log n \rceil^k$ indices, each index being a binary number of length $\lceil \log n \rceil$, followed by $c$ clauses. For every $1 \leq i \leq c$, the $i$-th index points to the first position in the $i$-th clause. Clauses start with $\#$ and are followed by a list of literals. Positive literals start with a $+$, negative with a $-$. The $+$ or $-$ symbol of a literal is followed by the ID of the variable in binary. $M$ proceeds as follows: (1) Using binary search with the aid of the ``out of range'' response $\triangleleft$, compute $n$ and $\lceil \log n \rceil$. (2) Copy the indices to a working tape, counting the number of indices (clauses) $c$. (3) Non-deterministically guess $c$ input addresses $a_1, \ldots, a_c$, i.e., guess $c$ binary numbers of length $\lceil \log n \rceil$. (4) Using $c$ $1$-bit flags, check that each $a_1, \ldots, a_c$ address falls in the range of a different clause. (5) Check that each $a_1, \ldots, a_c$ address points to an input symbol $+$ or $-$. (6) Copy the literals pointed by $a_1, \ldots, a_c$ to a working tape, checking that there are \emph{no} complementary literals. (7) Accept if all checks hold.
\end{example}

Let $L$ be a language accepted by a random-access Turing machine $M$. Assume that for some function $f$ on the natural numbers, $M$ makes at most $O(f(n))$ steps before accepting an input of length $n$. If $M$ is deterministic, then we write $L \in \dtime{f(n)}$. If $M$ is non-deterministic, then we write $L \in \ntime{f(n)}$. We define the classes of deterministic and non-deterministic poly-logarithmic time computable problems as follows:
\[\polylog = \bigcup_{k, c \in \mathbb{N}} \dtime{(\log n)^k \cdot c}\]
\[\mnpolylog = \bigcup_{k, c \in \mathbb{N}} \ntime{(\log n)^k \cdot c} \]
The non-deterministic random-access Turing machine in Example~\ref{ex:machine} clearly works in polylog-time. Therefore, polylogCNFSAT $\in \mnpolylog$.

Recall that an alternating Turing machine comes with a set of states $Q$ that is partitioned into subset $Q_\exists$ and $Q_\forall$ of so-called existential and universal states. Then a configuration $c$ is accepting iff
\begin{itemize}

\item $c$ is in a final accepting state,

\item $c$ is in an existential state and there exists a next accepting configuration, or

\item $c$ is in a universal state, there exists a next configuration and all next configurations are accepting.

\end{itemize}

In analogy to our definition above we can define a \emph{random-access alternating Turing machine}. The languages accepted by such a machine $M$, which starts in an existential state and makes at most $O(f(n))$ steps before accepting an input of length $n$ with at most $m$ alternations between existential and universal states, define the complexity class $\atime{f(n)}{m}$. Analogously, we define the complexity class $\atimeop{f(n)}{m}$ comprising languages that are accepted by a random-access alternating Turing machine that starts in a universal state and makes at most $O(f(n))$ steps before accepting an input of length $n$ with at most $m-1$ alternations between universal and existential states. With this we define
\[\polysigma{m} = \bigcup_{k, c \in \mathbb{N}} \mathrm{ATIME}[(\log n)^k \cdot c, m] \qquad \polypi{m} = \bigcup_{k,c \in \mathbb{N}} \mathrm{ATIME}^{op}[(\log n)^k \cdot c, m] . \]

The poly-logarithmic time hierarchy is then defined as $\mathrm{PLH} = \bigcup_{m \ge 1} \polysigma{m}$. Note that $\polysigma{1} = \mnpolylog$ holds. 

\begin{remark}

Note that a simulation of a $\mnpolylog$ Turing machine $M$ by a deterministic machine $N$ requires checking all computations in the tree of computations of $M$. As $M$ works in time $({\log n})^{O(1)}$, $N$ requires time $2^{{\log n}^{O(1)}}$. This implies $\mnpolylog \subseteq \mathrm{DTIME}(2^{{\log n}^{O(1)}})$, which is the complexity class called quasipolynomial time of the fastest known algorithm for graph isomorphism \cite{babai:stoc2016}, which further equals the class  
$\mathrm{DTIME}({n^{{\log n}^{O(1)}}})$\footnote{This relationship appears quite natural in view of the well known relationship $\mathrm{NP} = \mathrm{NTIME}(n^{O(1)}) \subseteq \mathrm{DTIME}(2^{{n}^{O(1)}}) = \mathrm{EXPTIME}$.}.
\end{remark}

\section{Proper Hierarchies in Polylogarithmic Time}\label{section:properhierarchies}

In this section we show that there are proper infinite hierarchy of problems inside each of the relevant polylogarithmic time complexity classes. We prove these facts constructively by means of problems consisting on deciding simple languages of binary strings. Our first results confirms that there is indeed a strict infinite hierarchy of problems in DPolylogTime.

\begin{theorem}\label{strictDet}
For every $k > 1$, $\dtime{\log^k n} \subsetneq \dtime{\log^{k+1} n}$.
\end{theorem}
\begin{proof}
Let $\initial{k}$ be the problem of deciding the language of all binary strings which have a prefix of at least $\lceil\log n\rceil^k$ consecutive zeros, where $n$ is the length of the string. 
For the upper bound, note that a random-access Turing machine can clearly check whether the first $\lceil \log n \rceil^{k+1}$ bits in the input-tape are $1$ by working in deterministic time $O(\log^{k+1} n)$. Thus $\initial{k+1} \in \dtime{\log^{k+1}n}$. 

Regarding the lower bound, we show that $\initial{k+1}$ does \emph{not} belong to $\dtime{\log^{k}n}$. 
Let us assume for the sake of contradiction that there is a deterministic random-access Turing machine $M$ that decides $\initial{k+1}$ in time $\lceil \log n \rceil^k \cdot c$, for some constant $c \geq 1$. Take a string $s$ of the form $0^n$ such that $\lceil \log n \rceil^{k+1} >  \lceil \log n \rceil^k \cdot c$. Since the running time of $M$ on input $s$ is strictly less than $\lceil\log n\rceil^{k+1}$, then there must be at least one position $i$ among the first $\lceil\log n\rceil^{k+1}$ cells in the input tape that was not read in the computation of $M(s)$. Define a string $s' = 0^i10^{n-i-1}$. Clearly, the output of the computations of $M(s)$ and $M(s')$ are identical. This contradicts the assumption that $M$ decides $\initial{k+1}$, since it is not true that the first $\lceil\log n\rceil^{k+1}$ bits of $s'$ are $0$.    
\qed
\end{proof}

Our second hierarchy theorem shows that there is also a strict hierarchy of problem inside \npolylog.

\begin{theorem}\label{strictNonD}
    For every $k > 1$,  $\ntime{\log^k n} \subsetneq \ntime{\log^{k+1} n}$.
\end{theorem}

\begin{proof}
Let $\conseq{k}$ denote the problem of deciding the language of binary strings which have at least $\lceil\log n\rceil^k$ consecutive bits set to $0$, where $n$ is the length of the string. 
For the upper bound we show that $\conseq{k+1}$ is in $\ntime{\log^{k+1}n}$. A random-access Turing machine $M$ can  guess (non-deterministically) a position $i$ in the input tape in time $O(\log n)$ by simply guessing $\lceil\log n\rceil$ bits and writing them in the address-tape. Then $M$ can check (working deterministically) in time $O(log^{k+1} n)$ whether each cell of the input tape between positions $i$ and $i+log^{k+1} n$ has a $0$. 

Regarding the lower bound, we need to show that $\conseq{k+1}$ is \emph{not} in $\ntime{\log^{k}n}$.
Let us assume for the sake of contradiction that there is a nondeterministic random-access Turing machine $M$ that decides $\conseq{k+1}$ in time $\lceil \log n \rceil^k \cdot c$, for some constant $c \geq 1$. Take a binary string $s$ of the form $0^{\lceil \log n \rceil^{k+1}}1^{n-\lceil \log n \rceil^{k+1}}$ such that $\lceil \log n \rceil^{k+1} >  \lceil \log n \rceil^k \cdot c$. Since $M$ accepts $s$, then there is at least one computation $\rho$ of $M$ which accepts $s$ in at most $\lceil \log n \rceil^k \cdot c$ steps. Then there must be at least one position $i$ among the first $\lceil \log n \rceil^{k+1}$ cells in the input tape that was not read during computation $\rho$. Define a string $s' = 0^{i}10^{\lceil \log n \rceil^{k+1}-i-1}1^{n-\lceil \log n \rceil^{k+1}}$. Clearly, the accepting computation $\rho$ of $M(s)$ is also an accepting computation of $M(s')$. This contradicts the assumption that $M$ decides $\conseq{k+1}$, since it is not true that there are $\lceil \log n \rceil^{k+1}$ consecutive zeros in $s'$. 
\qed
\end{proof}

Regarding the complement of $\npolylog$, the following theorem shows that there is a strict hierarchy of problems inside the first level of the $\tilde{\Pi}_m^{\mathit{plog}}$ hierarchy. 

\begin{theorem}\label{hierachyInPi0}
    For every $k > 1$,  $\atimeop{\log^k n}{1} \subsetneq \atimeop{\log^{k+1} n}{1}$.
\end{theorem}

\begin{proof}
Let $\noconseq{k}$ denote the problem of deciding the language of binary strings which do \emph{not} have greater than or equal $\lceil\log n\rceil^k$ consecutive bits set to $0$, where $n$ is the length of the string. For the upper bound we prove that  $\noconseq{k+1} \in \atimeop{\log^{k+1} n}{1}$. In a universal state, a random-access alternating Turing machine $M$ can check whether for all cell in some position $i$ in the input tape that is at distance at least $\lceil \log n \rceil^{k+1}$ from the end of the tape, there is a position between positions $i$ and $i+\lceil \log n \rceil^{k+1}$ with $1$. Each of those checking can be done deterministically in time $O(log^{k+1} n)$. Therefore this machine decides $\noconseq{k+1}$ in $\atimeop{\log^{k+1} n}{1}$.

Regarding the lower bound, we show that $\noconseq{k+1}$ is \emph{not} in $\atimeop{\log^k n}{1}$.
Let us assume for the sake of contradiction that there is an alternating random-access Turing machine $M$ that decides $\noconseq{k+1}$ using only universal states and in time $\lceil \log n \rceil^k \cdot c$, for some constant $c \geq 1$. Take a binary string $s$ of the form $0^{\lceil \log n \rceil^{k+1}}1^{n-\lceil \log n \rceil^{k+1}}$ such that $\lceil \log n \rceil^{k+1} > \lceil \log n \rceil^k \cdot c$. From our assumption that $M$ decides $\noconseq{k+1}$, we get that there is a rejecting computation $\rho$ of $M(s)$. Since every computation of $M$ which rejects $s$ must do so reading at most $\lceil \log n \rceil^k \cdot c$ cells, then there must be at least one position $i$ among the first $\lceil \log n \rceil^{k+1}$ cells in the input tape that was not read during computation $\rho$. Define a string $s' = 0^i10^{{\lceil \log n \rceil^{k+1}} - i - 1}1^{n-\lceil \log n \rceil^{k+1}}$. Clearly, the rejecting computation $\rho$ of $M(s)$ is also a rejecting computation of $M(s')$. This contradicts the assumption that $M$ decides $\noconseq{k+1}$, since $s'$ do not have $\lceil\log n\rceil^{k+1}$ consecutive bits set to $0$ and should then be accepted by all computations of $M$.   
\qed
\end{proof}

Next we show that there is also a strict hierarchy of problems inside the second level of the $\tilde{\Sigma}_m^{\mathit{plog}}$ hierarchy.

\begin{theorem}\label{Th:sigmatwo}
    For every $k > 1$,  $\atime{\log^k n}{2} \subsetneq \atime{\log^{k+1} n}{2}$.
\end{theorem}

\begin{proof}
Let $\exactlyonce{k}$ denote the problem of deciding the language of binary strings which contain the substring $0^{\lceil \log n \rceil^k}$ exactly once, i.e., $s$ is in $\exactlyonce{k}$ iff $0^{\lceil \log n \rceil^k}$ is a substring of $s$ and every other substring of $s$ is not $0^{\lceil \log n \rceil^k}$.
For the upper bound we show that  $\exactlyonce{k+1}$ is decidable in $\atime{\log^{k+1} n}{2}$ by combining the machines that decide $\conseq{k+1}$ and $\noconseq{k+1}$ in the proofs of Theorems~\ref{strictNonD} and~\ref{hierachyInPi0}, respectively. An alternating random-access Turing machine $M$ can decide $\exactlyonce{k+1}$ as follows: First $M$ checks in an existential state whether there is a position $i$ in the input tape such that each cell between positions $i$ and $i + \lceil \log n \rceil^{k+1}$ has a $0$. Then $M$ switches to a universal state and checks whether for all cell in some position $j$ that is at distance at least $\lceil \log n \rceil^{k+1}$ from the end of the input tape other than position $i$, there is a cell between positions $j$ and $j+\lceil\log n\rceil^{k+1}$ with $1$. If these two checks are successful, then the input string belongs to $\exactlyonce{k+1}$. We already saw in the proofs of Theorems~\ref{strictNonD} and~\ref{hierachyInPi0} that both checks can be done in time $O(\log^{k+1} n)$.

 Regarding the lower bound, we show that $\exactlyonce{k+1}$ is \emph{not} decidable in $\atime{\log^k n}{2}$.
We assume for the sake of contradiction that there is an alternating random-access Turing machine $M$ that decides $\exactlyonce{k+1}$ in $\atime{\log^k n}{2}$. We further assume, w.l.o.g., that every final state of $M$ is universal.
Let $M$ work in time $\lceil \log n \rceil^k \cdot c$ for some constant $c$. Take a binary string $s$ of the form $0^{\lceil \log n \rceil^{k+1}}10^{\lceil \log n \rceil^{k+1}}1^{n- 2 \cdot \lceil \log n \rceil^{k+1} -1}$ such that $\lceil \log n \rceil^{k+1} > \lceil \log n \rceil^k \cdot c$. From our assumption that $M$ decides $\exactlyonce{k+1}$, we get that there is a rejecting computation $\rho$ of $M(s)$.  Since every computation of $M$ which rejects $s$ must do so reading at most $\lceil \log n \rceil^k \cdot c$ cells, then there must be a position $i$ among the first $\lceil \log n \rceil^{k+1}$ cells in the input tape that was not read during computation $\rho$.
 Define a string \[s' = 0^i10^{{\lceil \log n \rceil^{k+1}} - i - 1}10^{\lceil \log n \rceil^{k+1}}1^{n- 2\cdot \lceil \log n \rceil^{k+1} -1}.\] Clearly, the rejecting computation $\rho$ of $M(s)$ is still  a rejecting computation of $M(s')$. This contradicts the assumption that $M$ decides $\exactlyonce{k+1}$, since $s'$ has exactly one substring $0^{\lceil\log n\rceil^{k+1}}$ and should then be accepted by all computations of $M$.   
\qed
\end{proof}

Together with Theorems~\ref{strictNonD} and~\ref{Th:sigmatwo} the following result shows that there is a proper hierarchy of problems for every level of the polylogarithmic time hierarchy $\polysigma{m}$. 

\begin{theorem}\label{ThForSigma}
For $m > 2$ and $k > 1$, it holds that \[\atime{\log^k n}{m} \subsetneq \atime{\log^{k+2} n}{m}.\]
\end{theorem}

\begin{proof}
Let $\genconseq{k}{l}$ (respectively $\genexactly{k}{l}$) for $k,l \geq 0$ denote the problems of deciding the language of binary strings with at least (respectively exactly) $(\lceil\log n\rceil^k)^l$ non-overlapping adjacent substrings of the form $0^{\lceil \log n \rceil^k}$, where $n$ is the length of the string. That is, $\genconseq{k}{l}$ is the language of binary strings which have at least $(\lceil \log n \rceil^k)^{l+1}$ consecutive bits set to $0$ and $\genexactly{k}{l}$ is the language of binary strings which contain the substring $0^{(\lceil \log n \rceil^k)^{l+1}}$ exactly once.
For the upper bound we note that $\genconseq{k}{l}$ and $\genexactly{k}{l}$ are in $\atime{\log^{k+1} n}{2 \cdot l +1}$ and $\atime{\log^{k+1} n}{2 \cdot l + 2}$, respectively. 
This follows from two facts: (a) As shown in Problems~4.5 and~4.6 in~\cite{FerrarottiGST20}, $\genconseq{k}{l}$ and $\genexactly{k}{l}$ can be expressed by formulae in the restricted fragments of second-order logic capturing, respectively, the levels $\polysigma{2 \cdot l + 1}$ and $\polysigma{2 \cdot l + 2}$ of the polylogarithmic time hierarchy. (b) A random-access turing machine can evaluate those formulae by guessing $\lceil \log n \rceil^k$  addresses, each of length $\lceil \log n \rceil$ (see Part~(a) in the proof of Theorem~6 in~\cite{FGST18}).
Since $m > 2$, we get that if $m$ is odd, then $\genconseq{k+1}{(m - 1)/2}$ is in $\atime{\log^{k+2} n}{m}$. Likewise, if $m$ is even, then $\genexactly{k+1}{(m - 2)/2}$ is in $\atime{\log^{k+2} n}{m}$. 

Regarding the lower bounds, it is easy to see (given our previous results in this section) that: (a) for odd $m$, $\genconseq{k+1}{(m - 1)/2}$ is \emph{not} in $\atime{\log^{k} n}{m}$, and (b) for even $m$, $\genexactly{k+1}{(m - 2)/2}$ is also \emph{not} in $\atime{\log^{k} n}{m}$. Note that if $m$ is odd, then we can prove (a) by contradiction following a similar argument than in the proof of the lower bound for Theorem~\ref{strictNonD}. Likewise, if $m$ is even, then we can prove (b) by contradiction following a similar argument than in the proof of Theorem~\ref{Th:sigmatwo}. 
\qed
\end{proof}

It is clear that by taking the complements of the problems $\genconseq{k}{l}$ and $\genexactly{k}{l}$, a similar result holds for each level of the $\piplog{m}$ hierarchy. 

\begin{theorem}\label{ThForPi}
For $m = 2$ and every $k > 1$, it holds that $\atimeop{\log^k n}{m} \subsetneq \atimeop{\log^{k+1} n}{m}$. Moreover, For every $m > 2$ and every $k > 1$, it holds that $\atimeop{\log^k n}{m} \subsetneq \atimeop{\log^{k+2} n}{m}$.
\end{theorem}

\section{Polylogarithmic-time and (Absence of) Complete Problems}\label{sec:complete}

In this section we show that none of the polylogarithmic time complexity classes studied in this paper have complete problems in the classical sense. 


We first note that to study complete problems for polylogarithmic time under m-reductions with sublinear time bounds does not make sense. Consider for instance \dpolylog reductions. 
Assume there is a complete problem $P$ for the class \npolylog under \dpolylog reductions. Let $P'$ belong to \npolylog and let $M$ be a deterministic random-access Turing machine that reduces $P'$ to $P$ in time $c' \cdot log^{k'} n$ for some constant $c'$. Then the output of $M$ given an instance of $P'$ of length $n$ has maximum length $c' \cdot \log^{k'} n$. This means that,  given an input of length $n$ for $P'$ and its reduction, the random-access Turing machine that computes the complete problem $P$ can actually compute $P(s)$ in time $O((\log \log n)^k)$ for some fixed $k$. This is already highly unlikely. If as one would expect there are more than a single complete problem for the class, then we could keep applying reductions from one problem to the other, infinitely reducing the time required to compute the original problem.     

Let us then consider the standard concept of Karp reducibility, i.e., deterministic polynomially bounded many-one reducibility, so that we can avoid the obvious problem described in the previous paragraph. Rather surprisingly, there is no complete problems for \dpolylog and \npolylog, even under these rather expensive reductions for the complexity classes at hand. 

\begin{theorem}
\dpolylog does \emph{not} have complete problems under deterministic polynomially bounded many-one reductions. 
\end{theorem}
\begin{proof}
We prove it by contradiction. Assume that there is such a complete problem $P$. Since $P$ is in \dpolylog, then there is a random-access Turing machine $M$ which computes $P$ in time $O(\log^k n)$ for some fixed $k$. Thus $P$ belongs to $\dtime{\log^k n}$. Let us take the problem $\initial{k+1}$ of deciding the language of binary strings which have a prefix of at least $\lceil\log n\rceil^{k+1}$ consecutive zeros. Since $P$ is complete for the whole class \dpolylog, there must be a function $f: \{0,1\}^* \rightarrow \{0,1\}^*$, computable in polynomial-time, such that $x \in \initial{k+1}$ iff $f(x) \in P$ holds for all $x \in \{0,1\}^*$. It then follows that the size of $f(x)$ is polynomial in the size of $x$. Let $|f(x)| = |x|^{k'}$, we get that the machine $M$ which computes the complete problem $P$ can also decide $\initial{k+1}$ in time $O(\log^{k} n^{k'}) = O((k' \cdot \log n)^k) = O(\log^k n)$. This contradicts the fact that $\initial{k+1} \not\in \dtime{\log^k n}$ as shown in the proof of Theorem~\ref{strictDet}.
\qed
\end{proof}

Using a similar proof strategy than in the previous theorem for \dpolylog, we can prove that the same holds for \npolylog. In fact, we only need to replace the problem $\initial{k+1}$ by $\conseq{k+1}$ and the reference to Theorem~\ref{strictDet} by a reference to Theorem~\ref{strictNonD} in the previous proof, adapting the argument accordingly. 

\begin{theorem}
\npolylog does \emph{not} have complete problems under deterministic polynomially bounded many-one reductions. 
\end{theorem}

Moreover, using the problems $\genconseq{k}{l}$ and $\genexactly{k}{l}$ together with its complements and Theorems~\ref{ThForSigma} and~\ref{ThForPi}, it is easy to prove that the same holds for every individual level of the polylogarithmic time hierarchy.

\begin{theorem}
For every $m \geq 1$, $\sigmaplog{m}$ and $\piplog{m}$ do \emph{not} have complete problems under deterministic polynomially bounded many-one reductions. 
\end{theorem}

\section{Direct-Access Turing Machines}\label{datm}

In this section we review the direct-access model of Turing machine introduced in \cite{FerrarottiGTBV19b,FerrarottiGTBV19}. We use this model to prove our result regarding a uniform complete problem for PolylogSpace. The definition of the model below have some small changes with respect to its definition in~\cite{FerrarottiGTBV19b,FerrarottiGTBV19}. These changes do not affect the key idea of accessing the different relations, functions and constants of the input structure directly through dedicated query tapes. Further, the model remains equivalent to the random-access model with respect to polylogarithmic time and space complexity classes. Regarding polylogarithmic space, it even remains equivalent to the standard Turing machine model.

Let $\sigma = \{R^{r_1}_1, \ldots, R^{r_p}_p, c_1,
\ldots c, f^{k_1}_1, \ldots, f^{k_s}_s\}$ be a vocabulary. A
\emph{direct-access Turing machine that takes $\sigma$-structures
$\mA$ as input}   
is a multitape Turing machine with:
\begin{itemize}
\item $p + s$ distinguished \emph{address-tapes} for relations and functions denoted as $ATR$ and $ATf$, respectively. 
\item $p + s$ distinguished read-only \emph{value-tapes} for relations and functions denoted as $VTR$ and $VTf$, respectively.
\item $l+1$ distinguished read-only \emph{constant-tapes} denoted as $CT$.
\item One or more ordinary \emph{work-tapes} denoted as $WT$.
\end{itemize}
The additional $CT$ (note that there are $l+1$ of them) holds the size $n$ of the domain of $\mA$

The set $Q$ of states of a direct-access Turing machine is assumed to have pairwise disjoint subsets $Q_{R_1}$, $\ldots$, $Q_{R_p}$, $Q_{f_1}$, $\ldots$, $Q_{f_s}$ of states, a unique accepting state $q_a$ and an initial state as $q_0$. 

The transition function $\delta$ is defined as usual. It takes 
  as input the current state of the machine and the values read by
  all tape-heads, and determines the new state and the values to be written in all tapes which are \emph{not} read-only.

The contents of the read only value-tapes is evaluated by the finite control of $M$ in $0$ time, at the beginning of a transition, before the transition function has been applied. This only happens if the (old) state $q$ is in the subset $Q_{f_i}$ (or $Q_{R_i}$) of the set $Q$ of states of the machine. This corresponds naturally to the idea that at the beginning of a transition in such a state, not only will in the $AT$ of  $f_i$ ($R_i$) be assumed to be stored the arguments of the function or relation, but also in the respective $VT$ the value of $f_i$ in those arguments will be stored (respectively, the Boolean value representing the fact that the $r_i$ tuple belongs to $R_i$ in $\mathbf{A}$).

If the state in the left side of the transition is \emph{not} in an state belonging to some $Q_{R_i}$ or $Q_{f_i}$ subset, the only allowed symbol for the corresponding $VT$ is $\sqcup$ (blank).
We do \emph{not} allow in the $AT$ any value which is not in the domain, i.e., any non zero value in the range  $[\log n - \lceil \log n \rceil]$, so that the blank symbol will only appear in $VT$ in the initial configuration (as in all the other tapes, except $CT$).

If $C$ is an accepting configuration of $M$ on a certain input, its successor configuration is $C$. That is, once that $M$ enters into the accepting state $q_a$, it remains in the same state, and all the tapes heads remain unchanged.

All the tapes of $M$, with the exception of the $VT$ of the relations symbols, have virtual \emph{end marks} as follows: the $AT$  of the relations symbols, the $AT$ and $VT$ of the function symbols, and the $CT$, have a special virtual mark $\alpha$ immediately before the first cell, and a special virtual mark $\omega$ immediately after the last cell. The $WT$ have only the mark $\alpha$ immediately before the first cell.

\begin{theorem}[\cite{FerrarottiGTBV19b,FerrarottiGTBV19}]	\label{directrandom}
A class of finite ordered
structures $\cal C$ of some fixed vocabulary $\sigma$ is decidable by a random-access machine working in
$\polylog$ with respect to $\hat{n}$,
where $\hat{n}$ is the size of the binary encoding of the input structure, iff $\cal C$ is decidable by a direct-access Turing machine in
$\polylog$ with respect to $n$,
where $n$ is the size of the domain of the input structure.
\end{theorem}

\section{The Complexity Class $\polylogspace$}\label{section:polylogspace}

Let $L(M)$ denote the class of structures of a given signature $\sigma$ accepted by
 a direct-access Turing machine $M$. We say that $L(M) \in \mathrm{DSPACE}[f(n)]$ if $M$ visits at
most $O(f(n))$ cells in each work-tape before accepting or rejecting an input
structure whose domain is of size $n$.
We define the class of all languages decidable by a deterministic
direct-access Turing machines in \emph{polylogarithmic space} as
follows:
\begin{equation*}
\polylogspace := \bigcup_{k \in \mathbb{N}}
\mathrm{DSPACE}[(\left\lceil\log n \right\rceil)^k].
\end{equation*}
Note that it is equivalent whether we define the class $\mathrm{PolylogSpace}$ by means of direct-access Turing machines or random-access Turing machines. Indeed, from Theorem~\ref{directrandom} and the fact that the (standard) binary encoding of a structure $\mA$ is of size polynomial with respect to the cardinality of its domain $A$, the following corollary is immediate.

\begin{corollary}[\cite{FerrarottiGTBV19b,FerrarottiGTBV19}]	\label{tba}
A class of finite ordered structures $\cal C$ of some fixed vocabulary $\sigma$ is decidable by a random-access Turing machine working in $\polylogspace$ with respect to $\hat{n}$,
where $\hat{n}$ is the size of the binary encoding of the input structure, iff $\cal C$ is decidable by a direct-access Turing machine in $\polylogspace$ with respect to $n$, where $n$ is the size of the domain of the input structure.
\end{corollary}

Moreover, in the context of $\polylogspace$, there is no need for random-access address-tape for the input; $\polylogspace$ defined with random-access Turing machines coincide with $\polylogspace$ defined with (standard) Turing machines that have sequential access to the input.

\begin{proposition}[\cite{FerrarottiGTBV19b,FerrarottiGTBV19}]
A class of finite ordered structures $\cal C$ of some fixed vocabulary $\sigma$ is decidable by a random-access machine working in $\polylogspace$ with respect to $\hat{n}$ iff $\cal C$ is decidable by a standard (sequential-access) Turing machine in $\polylogspace$ with respect to  $\hat{n}$, where $\hat{n}$ is the size of the binary encoding of the input structure.
\end{proposition}

\section{An Alternative (Uniform) Notion of Completeness.}{\label{sec:ucompl}}

Let us fix some necessary terminology. Let $\mathcal{M}$ be a countably infinite class of deterministic direct-access machines such that for every integer $k > 0$ there is exactly one direct-access machine $M^{k} \in \mathcal{M}$. We say that $\mathcal{M}$ is \textit{uniform}  if there is a deterministic Turing machine $M_{\mathcal{M}}$ which for every input $k \geq 0$ builds an encoding of the corresponding $M^{k} \in \mathcal{M}$. 
A \textit{structural language} or \textit{structural problem} is a countably infinite class of structures of a given finite signature which is closed under isomorphisms. Let $\sigma$ be a finite signature, $\mathrm{Str}[\sigma]$ denotes the class of all finite $\sigma$-structures.
Let $\mathcal{L}$ be a countably infinite class of structural languages of a same finite signature $\sigma$, we call $\mathcal{L}$ a \textit{problem family} or \textit{language family}. A problem family $\mathcal{L}$ is \textit{compatible} with a structural language $C$ if $\bigcup_{L_i \in {\cal L}} L_i = C$.

We proceed now to formally define the notion of uniform completeness discussed in the introduction. 

\begin{definition}\label{def:1}
We define uniform  decidability, reduction, hardness and completeness as follows:
\begin{itemize}
\item  Let $\mathcal{L}$ be a problem family and $\mathcal{M}$ be a uniform countably infinite class of deterministic direct-access machines.
$\mathcal{M}$ \textit{uniformly decides} $\mathcal{L}$ if for every $L_i \in {\cal L}$ there is an $M_j \in \mathcal{M}$ such that $M_j$ decides $L_i$.
\item Let $\mathcal{D}$ be a complexity class. A structural language $C$ \emph{is uniformly in $\mathcal{D}$ via a language family $\cal L$} if the following holds:
\begin{itemize}
\item $\mathcal{L}$ is compatible with $C$.
\item There is a uniform countably infinite class of deterministic direct-access machines $\cal M$ which uniformly decides $\mathcal{L}$.
\item Each machine in $\cal M$ belongs to $\cal D$.
\end{itemize}
\item  There is a \emph{uniform many-one} $\mathrm{P}$ \emph{reduction} from a structural language $L$ to a language family $\mathcal{L}$ (denoted $L \leq_{m}^{uP} \mathcal{L}$) if there is a $L_i \in \mathcal{L}$ and a deterministic transducer Turing machine $M_{L,L_i}$ in $\mathrm{P}$ which computes a function $f: \mathrm{Str}[\rho] \rightarrow \mathrm{Str}[\sigma]$ such that $\mA \in L$ iff $f(\mA) \in L_i$. Note that $M_{L,L_i}$ computes a classical Karp reduction $L \leq_{m}^{P} {L}_i$.
\item The structural language $C$ is \emph{uniformly hard for $\mathrm{PolylogSpace}$ under uniform many-one $\mathrm{P}$ reductions via a language family ${\cal L}$} if ${\cal L}$ is compatible with $C$ and $L_j \leq_{m}^{uP} \mathcal{L}$ holds for every structural language $L_j$ decidable in PolylogSpace. 
\item We say that $C$ is \emph{uniformly complete for $\mathrm{PolylogSpace}$ under uniform many-one $\mathrm{P}$ reductions via a language family ${\cal L}$} if it is uniformly hard for $\mathrm{PolylogSpace}$ under uniform many-one $\mathrm{P}$ reductions via $\cal L$ and further $C$ is uniformly in $\mathrm{PolylogSpace}$ via $\cal L$. 
\end{itemize}  
\end{definition}

In structural complexity, classical complete problems lead to some interesting consequences such as Corollary 3.19c in~\cite{BDG_95} which states that if a $\mathrm{PSPACE}$ complete problem under $\mathrm{P}$ (Karp) reductions is in $\mathrm{P}$, then $\mathrm{PSPACE} = \mathrm{P}$. The following lemma shows that our ``relaxed'' notion of uniform completeness still allow us to derive similar kind of results.

\begin{lemma}\label{PlspInP}
Let $C$ be uniformly complete for $\mathrm{PolylogSpace}$ under uniform many-one $\mathrm{P}$ reductions \emph{via} the problem family $\mathcal{L}$. 
If $C$ is also \textrm{uniformly} in $\mathrm{P}$ \emph{via} $\mathcal{L}$ then $\mathrm{PolylogSpace} \subseteq \mathrm{P}$.
\end{lemma}

\begin{proof} (Sketch)
Let $\mathcal{M}$ and $\mathcal{M}'$ be the classes of deterministic direct access machines that uniformly decide $\mathcal{L}$ witnessing the facts that $C$ is  uniformly in $\mathrm{PolylogSpace}$ and P, respectively.
Since we assume that $C$ is uniformly complete for $\mathrm{PolylogSpace}$ under uniform many-one $\mathrm{P}$ reductions \emph{via} the problem family $\mathcal{L}$, it follows from Definition~\ref{def:1} that for each structural language $L_i$ in $\mathrm{PolylogSpace}$
 there is a transducer Turing machine $M_{L_i,L_j} \in \mathrm{P}$ which reduces  $L_i$ to some $L_j$ in $\mathcal{L}$.
The fact that $C$ is uniformly in $\mathrm{P}$ implies by Definition~\ref{def:1} that there is a direct-access machines $M' \in {\cal M}'$  that decides $L_j$ in P. 
Then to decide $L_i$ we can build a deterministic direct-access machine $M''$ by assembling together $M_{L_i,L_j}$ and $M'$, redirecting the output of $M_{L_i,L_j}$ to a work tape and making $M'$ read its input from that work tape. As both machines are in $\mathrm{P}$, we get that $M''$ is also in $\mathrm{P}$. Moreover, we can construct a deterministic Turing machine that simulates the direct-access machine $M''$ and still works in $P$. That can be done using a strategy simmilar to the one  in the proofs of Theorem~1 and Proposition~1 in~\cite{FerrarottiGTBV19b}.
\qed
\end{proof}

The result in Lemma~\ref{PlspInP} should be interpreted in the light of the  following well known relationship between between deterministic space and time. \[\mathrm{PolylogSpace} \subseteq \mathrm{DTIME} \bigg(2^{\big(\lceil\log n \rceil^{O(1)}\big)}\bigg)\]
Note that this upper bound for PolylogSpace corresponds to the class known as Quasipolynomial Time (see~\cite{babai:stoc2016}).



\section{A (Uniform) Complete Problem for $\mathrm{PolylogSpace}$.}\label{S-plQSAT-C}

Our uniformly complete problem for PolylogSpace, namely the $\mathrm{QSAT}^{pl}$ problem, is inspiered by the well known PSPACE complete problem of satisfiability of  quantified Boolean sentences ($\mathrm{QSAT}$, aka $\mathrm{QBF}$ in~\cite{BDG_95}). Further, the strategy used in~\cite{BDG_95} to prove that $\mathrm{QSAT}$ is complete for $\mathrm{PSPACE}$ under $\mathrm{P}$ Karp reductions (see Theorems~3.29 and~2.27 as well as Lemmas~3.22, 3.27 and~3.28 in~\cite{BDG_95}) serve us as base for the corresponding strategy to prove that $\mathrm{QSAT}^{pl}$  is indeed uniformly complete for PolylogSpace.  

Let us then briefly recall the strategy used in~\cite{BDG_95} to prove that $\mathrm{QSAT}$ is complete for $\mathrm{PSPACE}$ under $\mathrm{P}$ Karp reductions.
Given a deterministic Turing machine $\mathrm{M_L}$ that decides a problem $L$ in $\mathrm{DSPACE}(n^{c})$ and an input string $x$ to $\mathrm{M_L}$, the strategy consists in building a quantified Boolean sentence $\mathrm{Accepted}_{M_L}(x)$ which is satisfiable only if the input string $x$ is accepted by $\mathrm{M_L}$. The formula is built by iterating $m = c_0 \cdot n^{c}$ times a sub-formula $\mathrm{Access}_{2^{m},{M_L}}(\alpha,\beta)$, which is true when $\alpha$ and $\beta$  are two vectors of $c_1 \cdot n^{c}$ Boolean variables which encode valid configurations of the  computation of $\mathrm{M_L}$ on input $x$, and such that the configuration $\beta$ is reachable from the configuration $\alpha$ in at most $2^{m}$ steps, where $c_0$ and $c_1$ are constants that depend on $\mathrm{M_L}$. Note that when the formula $\mathrm{Accepted}_{M_L}(x)$ is evaluated (by the Turing machine $\mathrm{M_{\mathrm{QBF}}}$ that decides $\mathrm{QBF}$) the sub-formula $\mathrm{Access}_{2^{j},{M_L}}(\alpha_j,\beta_j)$ needs to be evaluated  $2^m$ times, which is the maximum length of a computation  of  $\mathrm{M_L}$ on an input of length $n$.
 The number of alternations of quantifiers in  $\mathrm{Accepted}_{M_L}(x)$ is $2 \cdot m - 1$ and the number of Boolean variables is $(3 \cdot m + 2) \cdot (c_1 \cdot n^{c})$, which corresponds to $O(c_0 \cdot c_1 \cdot n^{2 \cdot c})$.


To evaluate the formula $\mathrm{Accepted}_{M_L}(x)$, $\mathrm{M_{\mathrm{QBF}}}$ uses a stack to implement the recursive execution of a function called $\mathrm{Eval}$. The depth of the stack
is essentially the number of quantifiers (i.e., of Boolean variables), plus the depth in the nesting of parenthesis of the quantifier free sub-formula (since also the logical connectives are evaluated with $\mathrm{Eval}$). In each entry, the stack records the configurations at the given stage and the truth value of the sub-formulas already evaluated. For that, the stack needs space polynomial in $n$.


In the case of $\mathrm{QSAT}^{pl}$, we make two main changes to the problem $\mathrm{QSAT}$. First we add a list of binary trees (represented as heaps, see below) as \emph{external constraints} to the input formula. Besides the classical connectives in $\{\vee, \wedge, \neg, \rightarrow\}$ we include a new \emph{constraint check} connective $\odot$. Second we restrict the number of quantifiers in the quantified Boolean sentences to be polylogarithmic in the size of the input. The model of computation also differs since we work with direct-access Turing machines. 

Having only polylogarithmically many quantifiers (and variables) essentially allows us to evaluate the quantified boolean sentences in polylogarithmic space instead of polynomial space. In our case the value of $m$ is $ c_0 \cdot \lceil\log n\rceil^{c}$ instead of $c_0 \cdot n^{c}$. We as well change slightly the strategy for the evaluation, requiring the input sentence to be in prenex normal form and using the function $\mathrm{Eval}$ only for the quantifiers.

We use external constraints as follows. During the construction of the formula $\mathrm{Accepted}_{M_L}(x)$ each generated instance of the sub-formula $\mathrm{Config}_{2^{j}}(\alpha_j)$ --needed to check that the vector of free Boolean variables $\alpha_j$ encodes a valid configuration of $\mathrm{M_L}$-- includes the connective $\odot$ listing the variables which represent the state of $\mathrm{M_L}$ in the corresponding configuration as well as the variables that correspond to the contents of the address- and value-tapes for all relations and functions in the input structure $\mA$. In the evaluation of $\mathrm{Accepted}_{M_L}(x)$, whenever the variables that encode the state in the configuration $\alpha_j$ correspond to a state in $\mathrm{M_L}$ where a particular relation or function is queried by the machine, the values assigned to those variables are checked against the heap that represents the corresponding relation or function in $\mA$


\begin{remark}\label{}
Note that the input to a direct-access Turing machine is not part of its configurations. The configuration of a direct-access machine only includes the size of the domain of the input structure $\mA$ (as the contents of $CT_{l+1}$) and the contents of the address- and value-tapes of all the relations and functions in $\mA$. In order to check whether a given tuple is in a given relation or to know the value of a certain function on a given tuple, we must first instruct the machine to write the tuple in the corresponding address tape and then to enter in the state in $Q_{r_i}$ or $Q_{f_i}$ for that particular relation or function, respectively. That is why we add the heaps as external constraints in the input to $\mathrm{QSAT}^{pl}$ as a way to represent the relations and functions of $\mA$.
It is note worthy that even using the classical Turing machine model to define $\mathrm{M_L}$ it is not possible to include the contents of the input tape in the configurations, since we would then need to use polynomial space in the machines which decide $\mathrm{QSAT}^{pl}$. This is so because the stack would then need polynomial space to be able to hold such configurations.
\end{remark}

\begin{definition}\label{def:2}

A \emph{quantified Boolean sentence with external constraints}, denoted as $\mathrm{QBF}^c$, is a word model\footnote{That is a model that encodes a finite string as defined in Section~6.2 in~\cite{EF95}} of signature \[\sigma^{qbf} = \{ \leq^{2}, Suc^{1 \rightarrow 1}, X^{1}, 1^{1}, 0^{1}, \exists^{1}, \forall^{1}, \wedge^{1}, \vee^{1} \neg^{1}, \Rightarrow^{1}, \odot^{1}, (^{1}, )^{1}, ,^{1}, *^{1}\}\] of the following form:

\begin{itemize}

\item \emph{Parameters:} If the $\mathrm{QBF}^c$ structure is the output of a uniform $\mathrm{P}$ reduction from a problem $L \in \mathrm{PolylogSpace}$, then there is a  list of $8$ binary strings separated by ``*''. They represent the constants $c$, $c_0$, $c_1$, $c_2$, $p$, $s$, $r_M$ and $e_M$ that depend on the direct-access machine $\mathrm{M_L}$ that decides $L$. Recall that $p$ and $s$ are the number of relation and function symbols in the input signature $\sigma$ of $\mathrm{M_L}$. $r_M$ and $e_M$ are the maximum arities among the relation and function symbols, respectively. Otherwise, these parameters do \emph{not} appear in the structure.

\item \emph{Formula:} A prenex quantified Boolean formula with connectives in the set $\{\vee, \wedge, \neg, \rightarrow\}$, no free-variables, no repetition of the variable indices in the quantifier prefix, no parenthesis in the quantifier prefix, and with the quantifier free part fully parenthesised. The variables are encoded as  $Xb_1\ldots b_l$, where   $b_1\ldots b_l$  is a binary string and $l$ is the minimum number of bits needed to enumerate all the variables in the formula. The formula may include the constraint check connective $\odot$ with the following syntax: if $\psi$ is a quantifier free Boolean formula with a set $X$ of free variables, then
$\big(\big(\psi\big) \odot \big((\bar{x}),(\bar{y}, \bar{z}),(\bar{v}, \bar{w})\big)\big)$,
where all variables in $\bar{x}$, $\bar{y}$, $\bar{z}$, $\bar{v}$ and $\bar{w}$ are in $X$, is a wff. The number of variables in $\bar{x}$  must be the same as the length of the binary strings in ``constraint control intervals'' (see next item). The remaining variables in $\odot$ are associated to ``External Constraints'' (see below). The semantics of $\odot$ is clarified in the proof of Lemma~\ref{PLQSC-k-inPlsp}.



\item \emph{Constraint control intervals:} An optional sequence of binary numbers which starts with ``*'' and finishes with ``**'', and where the numbers are separated by ``*''. The numbers in the sequence must appear in increasing order and must be of a same length. The sequence should include as many numbers as there are heaps, i.e., $p + s$ numbers (cf. with ``external constraints'').



\item \emph{External Constraints:} An optional sequence of full binary trees represented as heaps, i.e., in arrays following the order of a traversal of the trees by levels and from left to right.
    The heaps are separated by ``*'' and terminate with ``**''. The number of heaps must coincide with the number of binary numbers in the ``constraint control intervals''.
   As heaps represent full binary trees, their size must be $2^{m + 1} - 1$. This corresponds to a full binary tree of depth $m$. Heaps are related to the connective $\bigodot$ in the formula as follows.
The number of variables in $\bar{x}$ must be the same as the number of bits in each of the binary strings in ``constraint control intervals''.
 The first $p > 0$ heaps are of the same depth $d_p > 0$ and correspond to the variables $\bar{y}$. The remaining heaps (say $s > 0$) are of the same depth $d_s > 0$ and correspond to the variables $\bar{z}$ and $\bar{w}$. There are $p$ variables in $\bar{v}$. 
The number of variables in $\bar{y}$ must be $p \cdot d_p$. The number of variables in $\bar{z}$ and $\bar{w}$ must be, respectively, $s \cdot (d_s - h)$ and $s \cdot h$ for some integer $0 < h < d_s$.




\item \emph{Interdependency:} The ``external constraints'' are \emph{interdependent} with the ``constraint control intervals'' and the constraint check connective $\bigodot$ in the formula. Either the three of them are in the structure, or none of them are.


 \end{itemize}

\end{definition}

The uniformly complete PolylogSpace problem $\mathrm{QSAT}^{pl}$ is defined as follows.

\begin{definition}	\label{def:3}

Let $\mathrm{QSAT}^{pl}_{k}$ be the structural language formed by the set of finite structures $\mS$ of vocabulary $\sigma^{qbf}$ that are quantified Boolean sentences as per Definition~\ref{def:2} and that either satisfy property~(a) or~(b) below, where $\varphi$ is the formula encoded by the structure $\mS$ (i.e., a quantified Boolean sentence with optional external constraints), $\mathrm{Bvar}(\varphi)$ is the set of Boolean variables in $\varphi$ and $\hat{n}$ is the size of the domain of $\mS$.
\begin{enumerate}[a.]
\item $\mS$ has parameters $c$, $c_0$, $c_1$, $c_2$, $p$, $s$, $r_M$ and $e_M$ such that $k = (5 \cdot c \cdot c_0 \cdot c_1 \cdot c_2 \cdot p \cdot s \cdot r_M \cdot e_M)$ and has external constraints.  
\item $\mS$ has \emph{no} parameters, if present external constraints are satisfied, $\varphi$ is true and $|\mathrm{Bvar}(\varphi)|^{3} \leq \lceil\log \hat{n}\rceil^{k}$.
\end{enumerate}
The corresponding problem family $\mathcal{P}$ and problem language $\mathrm{QSAT}^{pl}$ are defined as $\{\mathrm{QSAT}^{pl}_{k}\}_{k \in \mathbb{N}}$ and $\bigcup_{k \in \mathbb{N}} \mathrm{QSAT}^{pl}_{k}$, respectively.

\end{definition}

Note that property (a) in Definition~\ref{def:3} corresponds to the case where the structure $\mS$ is the output of a uniform $\mathrm{P}$ reduction from a problem in $\mathrm{PolylogSpace}$. The listed parameters are described in the proof of Lemma~\ref{uHardPlsp}. Conversely, Property~(b) corresponds to the case where the structure $\mS$ is \emph{not}  the output of a uniform $\mathrm{P}$ reduction from a problem in $\mathrm{PolylogSpace}$.
In this case we have that $(|\mathrm{Bvar}(\varphi)|)^{3} \leq \lceil\log \hat{n}\rceil^k$  holds. This requirement could have instead been expressed in therms of the size of $\varphi$ or in terms of its quantifier free sub-formula, which could appear as more natural. We chose however to express it in terms of $\hat{n}$, i.e., in terms of the size of the domain of the structure, in order to make the use of external constraints optional for the general case.

\section{Uniform Completeness of the language $\mathrm{QSAT}^{pl}$.}\label{SS-plQSAT-C}

We first show that $\mathrm{QSAT}^{pl}$ is \emph{uniformly hard} for $\mathrm{PolylogSpace}$.

\begin{lemma}\label{uHardPlsp} 
Let the structural language $\mathrm{QSAT}^{pl}$ and the language family $\mathcal{P}$ be as in Definition~\ref{def:3}.
Then $\mathrm{QSAT}^{pl}$ is \emph{uniformly hard} for $\mathrm{PolylogSpace}$ under uniform many-one $\mathrm{P}$ reductions \emph{via} $\mathcal{P}$.
\end{lemma}

\begin{proof}
 By definition the language family $\mathcal{P}$ is compatible with the structural language $\mathrm{QSAT}^{pl}$.
 We need to prove that for each structural language $L \in \mathrm{PolylogSpace}$ of some signature $\sigma$, there is a uniform many-one $\mathrm{P}$ reduction from $L$ to $\mathcal{P}$.
 That is, we need to show that there is a language $\mathrm{QSAT}^{pl}_{k}$ in $\mathcal{P}$ and a transducer Turing machine  $M_{L,\mathrm{QSAT}^{pl}_{k}}$ which on input $\mA$ of signature $\sigma$ builds in $P$ time a $\sigma^{qbf}$ structure $f(\mA)$ such that $\mA \in L \Leftrightarrow f(\mA) \in \mathrm{QSAT}^{pl}_{k}$.


Let $L$ be a structural language in $\mathrm{PolylogSpace}$, decided by the direct-access machine $\mathrm{M_L}$ in $\mathrm{DSPACE}(\lceil\log n\rceil^{c})$.
Let $\mA$ be an input structure to $\mathrm{M_L}$ of signature $\sigma = \{R^{r_1}_1, \ldots, R^{r_p}_p, f^{e_1}_1, \ldots, f^{e_s}_s, c_1,
\ldots, c\}$ and size $n$. Let $\varphi$ denote the prenex quantified Boolean sentence encoded in $f(\mA)$ and $\phi$ denote its quantifier-free part.
%
We build a Turing machine $M_{L,\mathrm{QSAT}^{pl}_{k}}$ which computes the reduction from $L$ to $\mathrm{QSAT}^{pl}_{k}$ in $\mathrm{P}$ time.

The fact that
$\mA \in L  \Leftrightarrow \varphi$ is true is straightforward regarding the connectives $\{\vee, \wedge, \neg, \rightarrow\}$ since for the construction of $\varphi$  we follow essentially the same strategy as in~\cite{BDG_95}.

We need to show however that the address- and value-tapes of the relations and functions that appear in all the encoded ``query state'' configurations agree with the actual relations and functions in the input structure $\mA$. We also need to show that the $\sigma^{qbf}$ structure $f(\mA)$ satisfies the conditions in Definition~\ref{def:3} and can be built in polynomial time.

\paragraph*{1: Parameters.\,}
Recall that $c$ is the exponent in the space bound of $\mathrm{M_L}$.
We denote as $c_0$  the constant from the expression $(2^{c_0 \cdot s(n)})$ which gives an upper bound for the number of different configurations in a Turing machine with space bound $s(n)$ (see Theorem 3.29 and proof of Lemma 2.25 in~\cite{BDG_95}).
We denote as $c_1$ the number of bits needed to encode in binary each symbol in the alphabet of $\mathrm{M_L}$. For the address and value tapes we use the alphabet $\{0, 1, \diamond,\sqcup\}$, where $\diamond$ is used to indicate the position of the tapes head, and $\sqcup$ is the blank (as in the value tapes out of the corresponding query states). In the work tapes we might have a bigger alphabet.
We denote as $c_2$  the constant exponent of the polynomial which bounds the size (and the time for their construction) of each instance of the sub-formulas $\mathrm{Config}$, $\mathrm{Next}$, $\mathrm{Equal}$, $\mathrm{Initial}$, and $\mathrm{Accepts}$ that are used to build the formula $\mathrm{Accepted}_{M_L}(x)$ (which we call $\varphi$ here). See the explanation above and Lemma 3.22 in~\cite{BDG_95}.
$p$ and $s$ are the number of relation and function symbols in $\sigma$, and $r_M$ and $e_M$ are the maximum arities of the relation and function symbols there.
The \emph{sizes} of all those parameters in $f(\mA)$ is $O(1)$.


\paragraph{2: Quantifier prefix of $\varphi$.\,}
Upon inspecting the construction of the formula $\mathrm{Accepted}_{M_L}(x)$ (called $\varphi$ here) in Theorem 3.29, and Lemma 3.28 in~\cite{BDG_95}, we note that there are $(3\cdot c_0 \cdot  \lceil\log n\rceil^{c} + 2)$ vectors of Boolean variables, named as $\alpha_i, \beta_i, \gamma_i$, each one representing a configuration of $\mathrm{M_L}$, and hence having $O(c_1 \cdot  \lceil\log n\rceil^{c})$ Boolean variables.
Then we have a total of less than $(c_0 \cdot c_1 \cdot \lceil\log n\rceil^{2 \cdot c})$ Boolean variables in $\varphi$.
Considering the symbols in $\sigma^{qbf}$ (see Definition 2) needed to encode that amount of variables with their quantifiers, including the number of bits needed to encode the index for each such variable, we have that the size of the quantifier prefix of $\varphi$ in $f(\mA)$ is less than $(c \cdot c_0 \cdot c_1 \cdot \lceil\log n\rceil^{2 \cdot c + 1})$.  


\paragraph{3: Quantifier free sub-formula of $\varphi$.\,}
The size of each occurrence of each of the sub-formulas $\mathrm{Config}$, $\mathrm{Next}$, $\mathrm{Equal}$, $\mathrm{Initial}$, and $\mathrm{Accepts}$ (see above) is less than  $(2 \cdot c_1 \cdot \lceil\log n\rceil^{c})^{c_2}$, i. e., the maximum number of Boolean variables in each one of those sub-formulas raised to the exponent $c_2$ (see 1, above). The number of occurrences of those sub-formulas in the formula $\mathrm{Accepted}_{M_L}(x)$, is $(2 + 5 \cdot c_0 \cdot \lceil\log n\rceil^{c} + 4)$.
To the product of the two last expressions we must add the approximate number of occurrences of parenthesis of $O(2 \cdot c_0 \cdot \lceil\log n\rceil^{c})$, and of connectives $\{\wedge, \vee, \neg, \Rightarrow\}$, $O(2 \cdot c_0 \cdot \lceil\log n\rceil^{c})$.
So that the size of the quantifier free sub-formula, \emph{before adding} the connectives for the constraints $\odot$ is less than
$(c_0 \cdot c_1^{c_2} \cdot \lceil\log n\rceil^{c \cdot (c_2 + 1)})$.


\paragraph{Connective $\odot$ in $\varphi$.\,}
As we said above, when we build $\mathrm{Accepted}_{M_L}(x)$ (which we call $\varphi$ in the general $\sigma^{qbf}$ structure), every time that we generate an instance of the sub-formula $\mathrm{Config}_{2^{j}}(\alpha_j)$, we use the connective $\odot$ to list the variables which represent the state of $\mathrm{M_L}$ in that configuration, and also the variables that correspond to the contents of the address and value tapes for \emph{all} the relations and functions in the input structure $\mA$ to $\mathrm{M_L}$.
At every point of the computation where we would write a quantifier free sub-formula  $\psi$ which is \emph{an instance}  of the sub-formula $\mathrm{Config}_{2^{j}}(\alpha_j)$, we would write it in $\varphi$, instead,  as follows
%

\vspace{.2cm}
    $\bigg(\big(\psi\big) \odot \big((x_1,\ldots, x_l),(y_{111},\ldots, y_{11h},\ldots, y_{1r_11},\ldots, y_{1r_1h}, \ldots, y_{p11},\ldots, y_{p1h},$

\vspace{.2cm}

\noindent    $\ldots, y_{pr_p1},\ldots, y_{pr_ph}, z_{111},\ldots, z_{11h},\ldots, z_{1e_11},\ldots, z_{1e_1h},\ldots, z_{s11},$

\vspace{.2cm}

$\ldots, z_{s1h},\ldots, z_{se_s1},\ldots, z_{se_sh}),(v_{1},\ldots, v_{p}, w_{11},\ldots,w_{1h},\ldots, w_{s1},\ldots, w_{sh})\big)\bigg)$
$h = \lceil\log n\rceil$,
 $y_{ijk}$ is the $k$-th bit of the $j$-th component of a candidate tuple for the $i$-th relation in $\mA$,
 and similarly for $z_{ijk}$, regarding functions in $\mA$.
 $v_{i}$ is the answer ($0$ or $1$) to a query about the existence of the candidate tuple encoded in the corresponding variables $y_{ijk}$, in the $i$-th relation in   $\mA$,
and $w_{ij}$
is the $j$-th bit of the value of  the $i$-th function in $\mA$, for the tuple encoded in the corresponding sequence of variables $z$.



%
The variables in $\bar{x}$ encode the state, so that there are $\lceil\log |Q_L|\rceil$ such variables.
The variables in $\bar{y}$ and $\bar{z}$ encode the contents of $ATR$ and $ATf$ for all relation and function symbols $R$ and $f$, respectively in  $\sigma$.
We need less than $((p \cdot r_M \cdot \lceil\log n\rceil) + (s \cdot e_M \cdot \lceil\log n\rceil))$ of those variables.
Similarly, the variables in $\bar{v}$ and $\bar{w}$ encode the contents of $VTR$ and $VTf$. We need $(p + s \cdot \lceil\log n\rceil)$ of those variables.

There are in  $\phi$ $\;(6 + c_0 \cdot \lceil\log n\rceil^{c})$ occurrences of the sub-formula $\mathrm{Config}_{2^{j}}(\alpha_j)$ or other sub-formulas which include it and hence we have to add the $\odot$ connective  to them.
Then, considering that each variable  can be encoded in $(1 + c \cdot \lceil\log \log n\rceil)$ symbols in $\varphi$, the total size of the quantifier free sub-formula, \emph{including} the connectives for the constraints $\odot$ in $f(\mA)$ is less than
$(c \cdot c_0 \cdot c_1^{c_2}  \cdot p \cdot s \cdot r_M \cdot e_M \cdot \lceil\log n\rceil^{c \cdot (c_2 + 2) + 2})$.



\paragraph{4: Constraint control intervals.\,}
These are $p + s$ ordered binary numbers separated by ``*'', and terminating with ``**'' (see Definition 2). They are used in connection with the connective $\odot$ in $\varphi$ and the heaps to check the constraints, which as we said above in this case, where the $\sigma^{qbf}$ structure is the output of a reduction, they are used to check that the values given arbitrarily by different valuations to the contents of the address and value tapes of the relations and functions in the input structure $\mA$ in the corresponding configuration, match the corresponding relations or functions.
Recall that there are also $p + s$  heaps. During the evaluation of the formula  $\varphi$, in each occurrence of the connective $\odot$, the binary number $b$ formed by the Boolean values assigned by the current valuation to the variables which represent the state of $\mathrm{M_L}$ ($\bar{x}$, see 3, above) is checked against the constraint control intervals, so that the relative position of the first number which is greater or equal than $b$ indicates the relative position of the heap against which the address  (variables $\bar{y}$ and $\bar{z}$) and value tapes (variables $\bar{v}$ and $\bar{w}$) of the corresponding relation or function (whose relative position is also indicated by $b$) will be checked. That is,
the constraint control intervals represent the maximum values of the subsets of states of $\mathrm{M_L}$ in the order  $Q_{R_1}$, $\ldots$, $Q_{R_p}$, $Q_{f_1}$, $\ldots$, $Q_{f_s}$. Hence the size of the constraint control intervals in $f(\mA)$ is $O(1)$.



\paragraph{5: External Constraints. \;}
We explain how we build the external constraints and the connective $\bigodot$ in $\phi$.
Note that in the encoding of a configuration in the free variables of $\psi$, while for each cell in the address and value tapes of the relations and functions in $\mA$, we may  need to use more than one Boolean variable, in the particular case of the connective $\odot$ we only use one variable instead. This is because in these specific tapes, when the state of $\mathrm{M_L}$ is in the subset, say, $Q_{R_1}$ we may \emph{only} have the symbols $0$ or $1$ in each cell, and on the other hand when a $\sigma^{qbf}$ structure is not the output of a reduction, it makes more sense to match each Boolean variable in the right argument of the connective $\bigodot$ with one bit in  the paths from the roots in the heaps,  to their leaves  (see below). For that matter we encode the symbols in the alphabet of $\mathrm{M_L}$ in such a way that the rightmost bit is $0$ for the symbol $0$, and $1$ for the symbol $1$, and we encode that bit in the corresponding single variable.


\paragraph{5.1: Heaps.\,}
 We build $p + s$ heaps in the structure, that is, one for each relation symbol and one for each function symbol in the input $\mA$.
     The sizes of the heaps that represent the relations in $\mA$ are $2\cdot (2^{\lceil\log n\rceil})^{r_1} - 1$,..., $2\cdot (2^{\lceil\log n\rceil})^{r_p} - 1$, for $R_1$,..., $R_p$, respectively. And the sizes of the heaps that represent the functions are $2\cdot (2^{\lceil\log n\rceil})^{e_1 + 1} - 1$,..., $2\cdot (2^{\lceil\log n\rceil})^{e_s + 1} - 1$, for $f_1$,..., $f_s$, respectively.
The addition of $1$ to the arities of the functions in the exponents
 is due to the fact that each function $f_i$ is represented as a $(e_i + 1)$-ary relation in its heap. We take the first $e_i$ components from its address tape (variables $\bar{z}$ in $\odot$, see 3, above) and the  $(e_i + 1)$-th component from its value tape (variables $\bar{w}$).

In all the heaps the only cells whose contents are meaningful are those corresponding to the leaves of the trees, i.e., the last level of each tree.
For every relation and function, say relation $R_i$, each such cell corresponds to a single path of length  $(r_i \cdot \lceil\log n\rceil)$ and that path is determined by choosing at each level $(r_i \cdot \lceil\log n\rceil) \geq l \geq 0$, the left child if the $l$-th bit is $0$, and the right child if it is $1$. Recall that in a heap, if the cell number at level $l$ is $j\geq 1$, then the cell number of the left child is $2 \cdot j$, and that of the right child is $2 \cdot j + 1$.
    %
Note that we are representing all the trees as full, even if they may not be. However, the paths which correspond to leaves in a heap that are not present in the tree that would represent the actual relation, will never be used, since they correspond to tuples where some components have values greater than $n - 1$ (recall that the domain of $\mA$ is $\{0,\ldots, n - 1\}$).

Considering the sizes given above, the total size of the heaps
is less than $((p \cdot (2 \cdot (2^{\lceil\log n\rceil})^{r_M} - 1)) +
(s \cdot (2 \cdot (2^{\lceil\log n\rceil})^{e_M + 1} - 1)) + p + s)$.
Note that $2^{\lceil\log n\rceil}$ is $O(n)$, and it is easy to see that the constant multiplying $n$ is very small, since $2^{\lceil\log n\rceil} - 2^{\log n} < 2$.

So that we can say that the size of the External Constraints in $f(\mA)$ is less than $((p \cdot (2 \cdot n^{r_M} - 1)) +
(s \cdot (2 \cdot n^{e_M + 1} - 1)) + p + s)$.


\paragraph{6: Total size and Space bound.\;} As in 1, Definition 3, we denote as $\hat{n}$ the size of the domain of the  $\sigma^{qbf}$ structure $f(\mA)$.
Considering the sizes of the parts of $f(\mA)$ in 1 to 5, above, we have that

\vspace{.15cm}

\noindent $\hat{n} \leq$

\vspace{.1cm}

$[O(1)] + [(c \cdot c_0 \cdot c_1 \cdot \lceil\log n\rceil^{2 \cdot c + 1})] +$

\vspace{.1cm}

          $[(c \cdot c_0 \cdot c_1^{c_2}  \cdot p \cdot s \cdot r_M \cdot e_M \cdot \lceil\log n\rceil^{c \cdot (c_2 + 2) + 2})]
           + [O(1)]  +$

\vspace{.1cm}

           $[((p \cdot (2 \cdot n^{r_M} - 1)) +
(s \cdot (2 \cdot n^{e_M + 1} - 1)) + p + s)]$.

\vspace{.2cm}

Clearly the most significant term in the expression above is the fifth.
We have then that
$\hat{n} \leq  O \big((2 \cdot p + 2 \cdot s) \cdot   n^{\mathrm{Max} \{r_M, e_M + 1\}}\big)$, so that the structure $f(\mA)$ can be built in polynomial time.

On the other hand, it is clear from 3, 4 and 5 above  that each configuration encoded in a sequence of variables in the formula $\varphi$ (i.e., the sequences denoted as $\alpha_j$ in the instances of the sub-formulas $\mathrm{Config}_{2^{j}}(\alpha_j)$),  are considered as valid \emph{only}   if  the contents of the address and value tapes of the relations and functions in $\mA$ that appear in it \emph{agree} with the real relations and functions of the input $\mA$, in the appropriate states $Q_{R_1}$, $\ldots$, $Q_{R_p}$, $Q_{f_1}$, $\ldots$, $Q_{f_s}$ in the configuration.

Also, clearly $f(\mA)$ satisfies the conditions of Definition 3, making $k = (5 \cdot c \cdot c_0 \cdot c_1 \cdot c_2 \cdot p
\cdot s \cdot r_M \cdot e_M)$.

Note that giving that value to $k$ is what warranties us that effectively \emph{there is one particular} $k$ s.t. there is a Karp reduction from $L$ to $\mathrm{QSAT}^{pl}_k$. More precisely, the reduction is to
$\mathrm{QSAT}^{pl}_{(5 \cdot c \cdot c_0 \cdot c_1 \cdot c_2 \cdot p
\cdot s \cdot r_M \cdot e_M)}$  (see ``Discussion on the Parameter $k$'', below).    
\qed
\end{proof}

\vspace{.5cm}

Next we need to show that $\mathrm{QSAT}^{pl}$ is uniformly in PolylogSpace. We start by showing that each structural language $\mathrm{QSAT}^{pl}_{k}$ can be decided by a corresponding direct-access machine $M_k$ with space bounded by $\lceil\log \hat{n}\rceil^{k}$.


\begin{lemma}\label{PLQSC-k-inPlsp} 
%
%
For every $k > 0$, the structural language $\mathrm{QSAT}^{pl}_{k}$ from Definition~\ref{def:3} is in $\mathrm{DSPACE}(\lceil\log \hat{n}\rceil^{k})$, where $\hat{n}$ is the size of the domain of the input structure.

\end{lemma}

\begin{proof}
Let $k > 0$, let $\mS$ be a structure of the signature $\sigma^{qbf}$, let $\varphi$ be the prenex quantified Boolean formula in $\mS$ and let $\phi$ be its quantifier free sub-formula. We build a direct-access machine $M_k$ that decides $\mS \in \mathrm{QSAT}^{pl}_{k}$ working in $\mathrm{DSPACE}(\lceil\log \hat{n}\rceil^{k})$.


\paragraph*{1: Evaluation of the Formula.\,}
As in~\cite{BDG_95}, we use a function $\mathrm{Eval}$ to evaluate recursively the quantifiers in the prefix of $\varphi$. On reading the $i$-th quantifier in the quantifier prefix of the formula, $Q_ix_i$, for some $i > 0$, we call the same function $\mathrm{Eval}$ twice, to evaluate the sub-formula starting in the next quantifier, $Q_{i+1}$  with the current valuation of the preceding variables $x_1,\dots, x_{i-1}$, and with the Boolean values $0$ (False), and $1$ (True),  for $x_i$, then evaluating the disjunction or the conjunction of the returned truth values of the two calls, depending on $Q_i$ being $\exists$ or $\forall$, respectively, and finally returning the result.

Every time that we reach the last quantifier in the prefix, we have a full valuation for the variables in     the quantifier free sub-formula $\phi$. Then we evaluate it in $\mathrm{DLOGSPACE}$ as in~\cite{Buss87} (there the algorithm for the evaluation of the so called ``formulas in the wide sense'' --i.e. Boolean formulas with variables and a value assignment-- works in $\mathrm{ALOGTIME}$, which is known to be in $\mathrm{DLOGSPACE}$ --see Theorem 2.32 in~\cite{Immerman99} among other sources--). For the evaluation of the $\bigodot$ connective, we proceed as explained below, in ``Evaluation of the $\bigodot$ Connective''.

Every time that we must evaluate $\phi$, we read it from the input, using the address tapes and value tapes of the function $Suc$, and  the unary relations $X$, $0$, $1$, etc., in $\mS$, to find the next element in the formula, and then to know its corresponding symbol. We read the current valuation from the stack (see below).

To implement the function $\mathrm{Eval}$ we use a stack. In each entry we record the whole quantifier prefix up to the current one, with the following format for each quantifier: $Q_iX\bar{b},v_1,v_2,$ where  $\bar{b}$ is the index of the variable in binary ($i$), and $v_1$ is the truth value currently assigned to the variable $X\bar{b}$.
 As to $v_2$, it works as follows. When   the sub-formula that follows starting with $Q_{i+1}$, say $\delta_{i+1}$, is evaluated with the value $0$ for $x_i$, $v_2$ is a blank. When we get the truth value $w_1$ of the result of the evaluation of $\delta_{i+1}$ with $x_i = 0$, we change $v_2$ to the value $w_1$. At that point, we change $v_1$ of $x_i$ to $1$, and call $\mathrm{Eval}$ to evaluate $\delta_{i+1}$ again.
 Then, when we get the truth value $w_2$ of the result of the evaluation of $\delta_{i+1}$ with $x_i = 1$, we change $v_2$ to the disjunction or conjunction of the previous value of $v_2$ with the value $w_2$, depending on $Q_i$ being $\exists$ or $\forall$, respectively. That value is the one which will be returned after the call of $\mathrm{Eval}$ for the sub-formula which starts with $Q_iX\bar{b}$, and follows with $\delta_{i+1}$.

The \emph{depth} of the stack is  the number of Boolean variables in $\varphi$ (denoted as $|\mathrm{Bvar}(\varphi)|$, as in Definition 3), which is less than $(c_0 \cdot c_1 \cdot \lceil\log n\rceil^{2 \cdot c})$ (see 2, in the proof of Lemma~\ref{uHardPlsp}).
Given the explanation above, the size of \emph{each entry} is $|\mathrm{Bvar}(\varphi)| \cdot (7 + \log (|\mathrm{Bvar}(\varphi)|)$, which is less than
$(c_0 \cdot c_1 \cdot c \cdot \lceil\log n\rceil^{2 \cdot c + 1})$.

So that the \emph{total size} of the stack is less than $(c_0^{2} \cdot c_1^{2} \cdot c \cdot \lceil\log n\rceil^{4 \cdot c + 1})$.



\paragraph*{2: Evaluation of the $\bigodot$ Connective. \,}
Consider the expression for the $\bigodot$ Connective as in the  proof of Lemma~\ref{uHardPlsp}.
 Let $v$ be a  valuation for the quantifier free sub-formula $\phi$, let $\bar{x} = x_1,\ldots, x_l$, and let $v(x_1),\ldots, v(x_l)$  be the binary number   $b_1, \ldots, b_l$. We compare $b_1, \ldots, b_l$ with the ordered binary numbers in the ``constraint control intervals'' in $\mS$, until we find the first of them greater or equal to it. Suppose that that number is the $i$-th number in the list.
 If $i \leq p$, we check whether the valuation of the $i$-th subsequence of $d_p$ variables in $\bar{y}$ corresponds to the path in the $i$-th heap (see Definition~\ref{def:2}, and the explanation in ``heaps'', in the proof of Lemma~\ref{uHardPlsp}), that ends in a leaf which has the symbol $1$ if $v(v_i) = 1$, or $0$ if $v(v_i) = 0$.
 If $i = p + j$, for some $j > 0$, we check whether the valuation of the $j$-th subsequence of $d_s - h$ variables in $\bar{z}$ (where the number of variables in $\bar{w}$ is $s \cdot h$), followed by
 the valuation of the $j$-th subsequence of $h$ variables in $\bar{w}$, corresponds to the path in the $(p + j)$-th heap, that ends in a leaf which has the symbol $1$.
 If the check is correct, we evaluate $\big(\big(\psi\big) \odot \big((\bar{x}),(\bar{y}, \bar{z}),(\bar{v}, \bar{w})\big)\big)$ as $\psi \wedge \mathrm{True}$, otherwise we evaluate it as $\psi \wedge \mathrm{False}$.

This process can clearly be performed in a space smaller than the size of the stack, $(c_0^{2} \cdot c_1^{2} \cdot c \cdot \lceil\log n\rceil^{4 \cdot c + 1})$ (see above).



%
\paragraph*{3: Space Considerations. \,}
Note that the size of the stack is what determines the upper bound in the space needed by $M_k$.


Recall that in the case where the $\mathrm{QBF}^c$ structure is \emph{not} the output of a uniform $\mathrm{P}$ reduction from any problem in $\mathrm{PolylogSpace}$, according Definition~\ref{def:3},
each fragment $\mathrm{QSAT}^{pl}_{k}$ is the fragment of $\mathrm{QSAT}^{pl}$ where the proportion   $(|\mathrm{Bvar}(\varphi)|)^{3} \leq \lceil\log \hat{n}\rceil^k$  holds. Note that by our analysis in the last two paragraphs in~1 above, $(|\mathrm{Bvar}(\varphi)|)^{3}$ is 
 greater than the stack size, 
 so that also in this case the space will be enough to hold the stack.                      
\qed



%
\end{proof}






\subsubsection*{Discussion on the Parameter $k$}

When  $\mS$ \emph{is} the output of a uniform $\mathrm{P}$ reduction,
we must be sure that only \emph{one value} of $k$ will be enough to \emph{check all inputs} to  $M_{L}$. This is to comply with Definition~\ref{def:1}: there must be at least one language $\mathrm{QSAT}^{pl}_{k}$ in the family $\mathcal{P}$
   and a transducer Turing machine $M_{L,\mathrm{QSAT}^{pl}_{k}}$ which computes the Karp reduction $L \leq^{P}_{m} \mathrm{QSAT}^{pl}_{k}$.
That is why in Definition~\ref{def:3} we required that
   $k = (5 \cdot c \cdot c_0 \cdot c_1 \cdot c_2 \cdot p \cdot s \cdot r_M \cdot e_M)$ (see 1, in the proof of Lemma~\ref{uHardPlsp}, for the explanation of each parameter).
   Note that as all those parameters \emph{depend} on $L$ or  $M_{L,\mathrm{QSAT}^{pl}_{k}}$, and are fixed for all inputs to $L$,  in this way we can fulfil the definition of uniform many-one P reduction.

   Recall from the proof of Lemma~\ref{PLQSC-k-inPlsp} above, that the amount of space which the direct-access machine $M_k$ needs to decide the fragment $\mathrm{QSAT}^{pl}_{k}$ is less than $\big(c_0^{2} \cdot c_1^{2} \cdot c \cdot \lceil\log n\rceil^{4 \cdot c + 1}\big)$.
   We want to be sure that that amount of space is enough for \emph{all} values of $n$. Recall that the size of the input to $M_k$ is denoted by $\hat{n}$, and the input to $M_L$ by $n$  (that input which is transformed by $M_{L,\mathrm{QSAT}^{pl}_{k}}$ to a $\mathrm{QBF}^c$ structure of size polynomial in $n$, and the relationship between the two sizes is given by the expression in 6, in the proof of Lemma ~\ref{uHardPlsp}).
   Considering that expression,
   this means that the following relationship must hold for all values of $n$:

\vspace{.3cm}

   $\big(c_0^{2} \cdot c_1^{2} \cdot c \cdot \lceil\log n\rceil^{4 \cdot c + 1}\big) \; <  \; \big(\lceil\log \hat{n}\rceil\big)^{(5 \cdot c \cdot c_0 \cdot c_1 \cdot c_2 \cdot p \cdot s \cdot r_M \cdot e_M)}$

\vspace{.35cm}

\noindent As we saw in that lemma, the most significant term in the upper bound for $\hat{n}$ is $[((p \cdot (2 \cdot n^{r_M} - 1)) +
(s \cdot (2 \cdot n^{e_M + 1} - 1)) + p + s)]$, so that replacing it in the relationship above, we have

\vspace{.3cm}

$\big(c_0^{2} \cdot c_1^{2} \cdot c \cdot \lceil\log n\rceil^{4 \cdot c + 1}\big) \; <  \;$

\vspace{.2cm}

 $\bigg(\big\lceil\log \big((p \cdot (2 \cdot n^{r_M} - 1)) +
(s \cdot (2 \cdot n^{e_M + 1} - 1)) + p + s\big)  \big\rceil\bigg)^{(5 \cdot c \cdot c_0 \cdot c_1 \cdot c_2 \cdot p \cdot s \cdot r_M \cdot e_M)}$

\vspace{.35cm}

\noindent   which is, roughly, equivalent to

\vspace{.3cm}

$(c_0^{2} \cdot c_1^{2} \cdot c \cdot \lceil\log n\rceil^{4 \cdot c + 1}) \; < \;\bigg(\mathrm{Max \{r_M, e_M + 1\}} \cdot \lceil\log n\rceil\bigg)^{(5 \cdot c \cdot c_0 \cdot c_1 \cdot c_2 \cdot p \cdot s \cdot r_M \cdot e_M)}$

 \vspace{.35cm}

\noindent Clearly   that relationship holds for all values of $n$.

  Note that we chose to include in the value of $k$ \emph{all} the constants that affect in some way the value of $\hat{n}$, to have a safe upper bound (see 6, in the proof of Lemma~\ref{uHardPlsp}).
  In that choice, as also in the other calculations of sizes in this article, we have been using bounds which are not tight, and which most likely could be diminished with a more detailed analysis. However, our goal is to prove that for each $L \in \mathrm{PolylogSpace}$ there is a fragment $\mathrm{QSAT}^{pl}_{k}$ to which $L$ can be uniformly reduced, and for which there is a direct-access machine which can decide it in  the appropriate space bound, and for that matter our calculations suffice.

     With that exponent in the bound for the space in $M_k$ we are sure that \emph{all} the inputs to
   $M_{L,\mathrm{QSAT}^{pl}_{k}}$ will be reduced to structures which will be evaluated by  $M_k$, and the \emph{only reason} why a given input to $M_{L}$ can be reduced to an input to $M_k$ which \emph{is not} in the language $\mathrm{QSAT}^{pl}_{k}$ may  be because it \emph{is not} in the language $L$ either.

The following lemma completes the last part of the puzzle needed to show that $\mathrm{QSAT}^{pl}_{k}$ is indeed uniformly complete for PolylogSpace.




\begin{lemma}\label{uni-P-inPlsp} 
%

%
Let $\mathcal{M} = \{M_k\}_{k \in \mathbb{N}}$ be a countably infinite class of deterministic direct-access machines, where $M_k$ is the direct-access machine described in the proof of Lemma~\ref{PLQSC-k-inPlsp} that decides whether $\mathrm{QSAT}^{pl}_{k} \in \mathrm{DSPACE}(\lceil\log \hat{n}\rceil^{k})$.
Then the following holds:
\begin{enumerate}[a.]

\item $\mathcal{M}$ is uniform.

\item The structural language $\mathrm{QSAT}^{pl}$ is \emph{uniformly} in $\mathrm{PolylogSpace}$ \emph{via} the language family
$\mathcal{P}$.

\end{enumerate}
\end{lemma}

\begin{proof}
We prove (a.) by build a transducer Turing machine $M_{\mathcal{M}}$ which reads as input an integer $k > 0$ and builds in its output tape an encoding of a Turing machine $M_k$ that decides whether $\mathrm{QSAT}^{pl}_{k} \in \mathrm{DSPACE}(\lceil\log \hat{n}\rceil^{k})$. 


The machine $M_k$ built by $M_{\mathcal{M}}$ works exactly as in the description in the proof of Lemma~\ref{PLQSC-k-inPlsp} except for the following added details:

\begin{itemize}

\item[a.1] At the beginning of every computation $M_k$ works in the following way:
 (i) it writes in its work tape $WT_1$ the value of $k$ in binary, and leaves the head pointing to the first cell;
(ii) it reads the size of the input structure $\hat{n}$ in its constant tape $CT_{l+1}$;
(iii) it writes $0$ in exactly the first $k \cdot \lceil\log \hat{n}\rceil$ cells in the work tape $WT_2$; 
(iv) in every work tape, except $WT_1$ and $WT_2$, it counts $\lceil\log \hat{n}\rceil^{k}$ cells and then in the next cell writes ``*'' (note that $\lceil\log \hat{n}\rceil^{k}$ is space constructible for any $k$); for the count it uses $WT_2$, counting in base $\lceil\log \hat{n}\rceil$, and then writes $0$ in all those cells.

\item[a.2] During its computation on any input, $M_k$ works as follows: (i) it does not write \emph{any} other data in $WT_1$, i.e., the sole purpose of that tape is to hold the value of the parameter $k$;
(ii) it uses $WT_2$ only as a counter up to $\lceil\log \hat{n}\rceil^{k}$;
(iii) whenever the machine needs to execute a loop whose bound is any function of $k$, it reads $k$ from $WT_1$, i.e., the value of $k$ is \emph{not} hard-wired in the finite control of $M_k$;
(iv) it clocks its use of work space using the marks written in the work tapes in 1.4 above: if in any $WT_i$, with $i > 2$, $M_k$ reads ``*'', it stops rejecting.

\end{itemize}

%

Note that the transition function of each $M_k$, except for the part described in a.1 above, is the \emph{same} for all values of $k$ and hard-wired in the finite control of $M_{\mathcal{M}}$.

%
Regarding part~(b), we get from~(a) and from the assumption in the lemma that $\mathcal{M}$ uniformly decides $\mathcal{P}$. By Definition~\ref{def:3}, the language family $\mathcal{P}$ is clearly compatible with the structural language $\mathrm{QSAT}^{pl}$. By Lemma~\ref{PLQSC-k-inPlsp}, each structural language $\mathrm{QSAT}^{pl}_{k}$ in $\mathcal{P}$ is in $\mathrm{DSPACE}(\lceil\log \hat{n}\rceil^{k})$, and hence in $\mathrm{PolylogSpace}$.
 
 Note that the uniformity of the language family $\mathcal{P}$ assures us that \emph{for every} $k > 0$, there is at least one direct-access machine  which decides $\mathrm{QSAT}^{pl}_{k}$ in $\mathrm{DSPACE}(\lceil\log \hat{n}\rceil^{k})$, and we can build the encoding of that machine.

So that the structural language $\mathrm{QSAT}^{pl}$ is uniformly in $\mathrm{PolylogSpace}$ via the language family $\mathcal{P}$.
 %
\qed



%
\end{proof}


From Lemmas~\ref{uHardPlsp} and~\ref{uni-P-inPlsp}, the following result is immediate.


\begin{theorem}
Let the structural language $\mathrm{QSAT}^{pl}$ and the language family $\mathcal{P}$ be defined as Definition~\ref{def:3}.
Then $\mathrm{QSAT}^{pl}$ is uniformly complete for $\mathrm{PolylogSpace}$ under uniform many-one $\mathrm{P}$ reductions \emph{via} $\mathcal{P}$.
\end{theorem}\label{}  
%

%

\section{Concluding Remarks} \label{sec:schluss}

In the first part of the paper we have seen that in the classical sense and under Karp reductions none of the classes DPolylogTime, NPolylogTime, $\polysigma{m}$ and $\polypi{m}$ ($m \in \mathbb{N}$) has a complete problem.
 This result follows from the existence of proper hierarchies inside each of the classes. The proof that such hierarchies exist is constructive by defining concrete problems parameterized by $k \in \mathbb{N}$ for each class. 
We expect that these results can be taken further towards an investigation of the strictness of the polylogarithmic time hierarchy as such. We also expect that similar strict hierarchies can be defined in terms of subsets of formulae in the fragments $\sigmaplog{m}$ and $\piplog{m}$ of the restricted second-order logic capturing the polylogarithmic time hierachy. Notice that the latter does not follow directly from the strict hierarchies proven in this paper, since in the proofs of the characterization results for the polylogarithmic-time hierarchy~\cite{FerrarottiGST18,FerrarottiGST19}, there is \emph{no} exact correspondence between the exponents in the polylogarithmic functions that bound the time complexity of the machines and the exponents in the restricted second-order variables of the $\soplog$ formulae that define the machines.

The second and final part of the paper explores an alternative notion of completeness for PolylogSpace which is inspired by the concept of uniformity from circuit complexity theory. We were then able to prove that we can still isolate the most difficult problems inside PolylogSpace and draw some of the usual interesting conclusions entailed by the  classical notion of complete problem (see in particular Lemma~\ref{PlspInP} and its corresponding discussion). This is relevant since it is well known since a long time that PolylogSpace has no complete problem in the usual sense. It is clear that this new concept of (uniform) completeness can be applied to all polylogarithmic time complexity classes considered in this paper. The identification of corresponding uniformly complete problems for each of those classes is left for future work.


%
%

\bibliographystyle{spmpsci}      
\bibliography{biblio}   

\end{document}